\documentclass[12pt]{article}

\usepackage[margin=1in]{geometry}
\usepackage{CJK, hyperref, amsmath, amssymb, amsthm, authblk}

\newtheorem{lemma}{Lemma}
\newtheorem{theorem}{Theorem}

\theoremstyle{definition}
\newtheorem{condition}{Condition}
\newtheorem{definition}{Definition}

\DeclareMathOperator{\tr}{tr}
\DeclareMathOperator*{\E}{\mathbb E}

\begin{document}

\begin{CJK*}{UTF8}{}

\title{Dynamics of R\'enyi entanglement entropy in diffusive qudit systems}

\CJKfamily{gbsn}

\author{Yichen Huang (黄溢辰)\thanks{yichuang@mit.edu}}

\affil{Center for Theoretical Physics, Massachusetts Institute of Technology, Cambridge, Massachusetts 02139, USA}

\maketitle

\end{CJK*}

\begin{abstract}

My previous work [arXiv:1902.00977] studied the dynamics of R\'enyi entanglement entropy $R_\alpha$ in local quantum circuits with charge conservation. Initializing the system in a random product state, it was proved that $R_\alpha$ with R\'enyi index $\alpha>1$ grows no faster than ``diffusively'' (up to a sublogarithmic correction) if charge transport is not faster than diffusive. The proof was given only for qubit or spin-$1/2$ systems. In this note, I extend the proof to qudit systems, i.e., spin systems with local dimension $d\ge2$.

\end{abstract}

\section{Introduction}

My previous work \cite{Hua19} studied the dynamics of R\'enyi entanglement entropy $R_\alpha$ in local quantum circuits with charge conservation. Initializing the system in a random product state, it was proved that $R_\alpha$ with R\'enyi index $\alpha>1$ grows no faster than ``diffusively'' (up to a sublogarithmic correction) if charge transport is not faster than diffusive.

For simplicity, Ref. \cite{Hua19} only gave a proof for qubit or spin-$1/2$ systems. While the proof also works for qudit systems (i.e., spin systems with local dimension $d\ge2$), this extension was not explicitly presented in Ref. \cite{Hua19}. Due to the recent interest \cite{Zni20}, in this note I give an exposition so that readers who are only interested in the results do not have to spend time verifying that every step of the proof in Ref. \cite{Hua19} remains valid for qudit systems. This note does not contain any essentially new ideas beyond those in Ref. \cite{Hua19}.

For completeness and for the convenience of the reader, definitions and proofs are presented in full so that this note is technically self-contained, although this leads to substantial text overlap with the original paper \cite{Hua19}. It is not necessary to consult Ref. \cite{Hua19} before or during reading this note. However, in this note I do not discuss the conceptual aspects of the work. Such discussions are in Ref. \cite{Hua19}.

I recommend related papers \cite{RPv19, ZL20, Zni20}, which study the same problem with various analytical and numerical methods. These works provide insights that are complementary to those in Ref. \cite{Hua19} and here.

The rest of this note is organized as follows. In Section \ref{s:pre}, I give basic definitions. In Section \ref{s:qudit}, I present results with a complete proof for qudit systems.

\section{Preliminaries} \label{s:pre}

Throughout this note, standard asymptotic notations are used extensively. Let $f,g:\mathbb R^+\to\mathbb R^+$ be two functions. One writes $f(x)=O(g(x))$ if and only if there exist constants $M,x_0>0$ such that $f(x)\le Mg(x)$ for all $x>x_0$; $f(x)=\Omega(g(x))$ if and only if there exist constants $M,x_0>0$ such that $f(x)\ge Mg(x)$ for all $x>x_0$; $f(x)=\Theta(g(x))$ if and only if there exist constants $M_1,M_2,x_0>0$ such that $M_1g(x)\le f(x)\le M_2g(x)$ for all $x>x_0$.

For notational simplicity, we do not specify the base of the logarithm explicitly. All equations involving logarithms are valid as long as the base is an arbitrary fixed positive number.

\begin{definition} [entanglement entropy]
The R\'enyi entanglement entropy $R_\alpha$ with index $\alpha\in(0,1)\cup(1,+\infty)$ of a bipartite pure state $\rho_{AB}$ is defined as
\begin{equation}
R_\alpha(\rho_A):=\frac{1}{1-\alpha}\log\tr(\rho_A^\alpha)=\frac{1}{1-\alpha}\log\sum_{i\ge1}\Lambda_i^\alpha,
\end{equation}
where $\Lambda_1\ge\Lambda_2\ge\cdots\ge0$ with $\sum_{i\ge1}\Lambda_i=1$ are the eigenvalues (in descending order) of the reduced density matrix $\rho_A=\tr_B\rho_{AB}$. The min-entropy is defined as
\begin{equation}
R_\infty(\rho_A):=\lim_{\alpha\to+\infty}R_\alpha(\rho_A)=-\log\Lambda_1.
\end{equation}
The von Neumann entanglement entropy is given by
\begin{equation}
\lim_{\alpha\to1}R_\alpha(\rho_A)=-\tr(\rho_A\log\rho_A).
\end{equation}
\end{definition}

\begin{lemma} \label{l:1}
For $\alpha>1$,
\begin{equation}
R_\infty(\rho_A)\le R_\alpha(\rho_A)\le\frac{\alpha}{\alpha-1}R_\infty(\rho_A).
\end{equation}
\end{lemma}

\begin{proof}
For completeness, we give a proof of this well-known result. The first inequality is due to the fact that $R_\alpha$ is monotonically non-increasing in $\alpha$ (this is why $R_\infty$ is called min-entropy). The second inequality follows from
\begin{equation}
R_\alpha(\rho_A)=\frac{1}{1-\alpha}\log\sum_{i\ge1}\Lambda_i^\alpha\le\frac{1}{1-\alpha}\log(\Lambda_1^\alpha)=\frac{\alpha}{\alpha-1}R_\infty(\rho_A).
\end{equation}
\end{proof}

\begin{definition} [local quantum circuit with charge conservation]
Consider a chain of $N$ qudits or spins with local dimension $d\ge2$. Assume without loss of generality that $N$ is even. Let the time-evolution operator be
\begin{equation}
U(t,0)=U(t,t-1)U(t-1,t-2)\cdots U(1,0),\quad t\in\mathbb Z^+.
\end{equation}
Each layer of the circuit consists of two sublayers of local unitaries:
\begin{equation} \label{lu}
U(t,t-1)=\prod_{i=1}^{N/2-1} U_{2i,2i+1}^{(t)}\times\prod_{i=1}^{N/2}U_{2i-1,2i}^{(t)}.
\end{equation}
Each unitary $U_{i,i+1}^{(t)}$ acts on two neighboring spins at sites $i,i+1$, and commutes with $S_i^z+S_{i+1}^z$, where $S_i^z$ is the $z$ component of the spin operator at site $i$. It should be clear that every $U_{i,i+1}^{(t)}$ and hence $U(t,0)$ preserve charge or the $z$ component $\sum_{i=1}^NS_i^z$ of the total spin.
\end{definition}

Let us consider some examples. For spin-$1/2$ ($d=2$),
\begin{equation}
S_i^z=\frac{1}{2}
\begin{pmatrix}
1 & 0 \\
0 & -1
\end{pmatrix}
\end{equation}
in the computational basis $\{|0\rangle,|1\rangle\}$, and
\begin{equation}
U_{i,i+1}^{(t)}=
\begin{pmatrix}
* & 0 & 0 & 0\\
0 & * & * & 0\\
0 & * & * & 0\\
0 & 0 & 0 & *
\end{pmatrix}
\end{equation}
is block diagonal in the basis $\{|00\rangle,|01\rangle,|10\rangle,|11\rangle\}$, where ``$*$'' denotes a possibly non-zero entry, i.e., $U_{i,i+1}^{(t)}$ is the direct sum of a phase factor, a unitary matrix of order $2$, and a phase factor.

For spin-$1$ ($d=3$),
\begin{equation}
S_i^z=
\begin{pmatrix}
1 & 0 & 0\\
0 & 0 & 0\\
0 & 0 & -1
\end{pmatrix}
\end{equation}
in the basis $\{|0\rangle,|1\rangle,|2\rangle\}$, and
\begin{equation}
U_{i,i+1}^{(t)}=
\begin{pmatrix}
* & 0 & 0 & 0 & 0 & 0 & 0 & 0 & 0\\
0 & * & * & 0 & 0 & 0 & 0 & 0 & 0\\
0 & * & * & 0 & 0 & 0 & 0 & 0 & 0\\
0 & 0 & 0 & * & * & * & 0 & 0 & 0\\
0 & 0 & 0 & * & * & * & 0 & 0 & 0\\
0 & 0 & 0 & * & * & * & 0 & 0 & 0\\
0 & 0 & 0 & 0 & 0 & 0 & * & * & 0\\
0 & 0 & 0 & 0 & 0 & 0 & * & * & 0\\
0 & 0 & 0 & 0 & 0 & 0 & 0 & 0 & *
\end{pmatrix}
\end{equation}
is block diagonal in the basis $\{|00\rangle,|01\rangle,|10\rangle,|02\rangle,|11\rangle,|20\rangle,|12\rangle,|21\rangle,|22\rangle\}$.

\begin{definition} [Haar-random local quantum circuit with charge conservation] \label{def:Hr}
Recall that each local unitary $U_{i,i+1}^{(t)}$ in Eq. (\ref{lu}) commutes with $S_i^z+S_{i+1}^z$ and is therefore block diagonal in the computational basis. The ensemble of Haar-random local quantum circuits with charge conservation is defined by letting each block of each $U_{i,i+1}^{(t)}$ be an independent Haar-random unitary.
\end{definition}

\section{Results and proofs} \label{s:qudit}

Recall that in our notation, the eigenstates of $S_i^z$ are $\{|0\rangle_i,|1\rangle_i,\ldots,|d-1\rangle_i\}$ with $S_i^z|k\rangle_i=((d-1)/2-k)|k\rangle_i$. Let $Q_i:=(d-1)/2-S_i^z$ be the charge operator. Since $Q_i|k\rangle_i=k|k\rangle_i$, we interpret $|k\rangle_i$ as there being $k$ charges on site $i$. Let $X_i$ be the generalized Pauli $X$ operator at site $i$ defined by
\begin{equation}
    X_i|k\rangle_i=|(k+1)\bmod d\rangle_i.
\end{equation}
Let $\{|0)_i,|1)_i,
\ldots,|d-1)_i\}$ be eigenstates of $X_i$. It is not difficult to see that $|\langle k|k')|=1/\sqrt d$ for any $k,k'=0,1,\ldots,d-1$, where the subscript $i$ has been dropped for notational simplicity.

Diffusive transport means that the transport of conserved quantities satisfies the diffusion equation at large distance and time scales. It can be considered, e.g., in the linear response regime and in quantum quench, where the system is infinitely close to and far from equilibrium, respectively. However, it is not clear whether diffusive transport in one setting is equivalent to or implies that in another (it might be possible that transport is diffusive in one setting but not in another). Our results rely on the following necessary condition for no-faster-than-diffusive transport.

\begin{condition} \label{def:diff}
Consider a chain of $N$ qudits divided into two subsystems $C\otimes D$. Subsystem $C$ is a contiguous region of $m$ qudits, and subsystem $D$ is the rest of the system. We initialize $C$ in the state $|0\rangle^{\otimes m}$ and $D$ in an arbitrary product state, i.e., each qudit in $D$ is disentangled from all other qudits. Let $i$ be the position of a qudit in the bulk of $C$ such that the distances from site $i$ to the two endpoints of $C$ are both $\Theta(m)$. Then,
\begin{equation} \label{eq:diff}
\langle\psi(t)|Q_i|\psi(t)\rangle=e^{-\Omega(m^2/t)},
\end{equation}
where $\psi(t)$ is the state (wave function) at time $t$.
\end{condition}

Equation (\ref{eq:diff}) implies that
\begin{equation} \label{eq:diffc}
\|(1-|0\rangle_i\langle0|_i)|\psi(t)\rangle\|^2=\langle\psi(t)|(1-|0\rangle_i\langle0|_i)|\psi(t)\rangle\le\langle\psi(t)|Q_i|\psi(t)\rangle=e^{-\Omega(m^2/t)}.
\end{equation}

Equation (\ref{eq:diff}) can be intuitively understood as follows. At initialization $t=0$, there is no charge in $C$, i.e., $C$ is in the all-zero state. Any charge observed on site $i$ at a later time $t>0$ must be transported from $D$ all the way to the bulk of $C$. The distance is $\Theta(m)$. The left-hand side of Eq. (\ref{eq:diff}) is the amount of charge on site $i$ at time $t$, and the right-hand side follows from the diffusion equation. In particular, a nonvanishing amount of charge requires that $t=\Omega(m^2)$. It should be clear that violating Eq. (\ref{eq:diff}) unambiguously implies that charge transport is faster than diffusive.

As an instructive example, we show that

\begin{lemma} \label{l:r}
For any initial state $|\psi(0)\rangle$ with no charge in $C$,
\begin{equation} \label{eq:rand}
    \Pr_{U(t,0)\in\mathcal R}\left(\langle\psi(0)|U^\dag(t,0)Q_iU(t,0)|\psi(0)\rangle= e^{-\Omega(m^2/t)}\right)=1-e^{-\Omega(m^2/t)},
\end{equation}
where $\mathcal R$ is the ensemble of Haar-random local quantum circuits with charge conservation (Definition \ref{def:Hr}).
\end{lemma}

\begin{proof}
It is not difficult to prove that the distribution of charge
\begin{equation}
    \left\{\E_{U(t,0)\in\mathcal R}\langle\psi(0)|U^\dag(t,0)Q_iU(t,0)|\psi(0)\rangle\right\}_{i=1}^{N}
\end{equation}
after averaging over the ensemble $\mathcal R$ evolves as an unbiased discrete random walk \cite{KVH18, RPv18}. Hence,
\begin{equation}
    \E_{U(t,0)\in\mathcal R}\langle\psi(0)|U^\dag(t,0)Q_iU(t,0)|\psi(0)\rangle=e^{-\Omega(m^2/t)}
\end{equation}
if site $i$ is in the bulk of $C$. Then, Eq. (\ref{eq:rand}) follows from Markov's inequality.
\end{proof}

We are ready to state and prove the main result.

\begin{theorem} \label{thm}
Consider a chain of $N$ qudits as a bipartite quantum system $A\otimes B$. Assume without loss of generality that $N$ is even. Subsystem $A$ consists of qudits at sites $1,2,\ldots,N/2$, and subsystem $B$ is the rest of the system (we study the entanglement across the middle cut). Initialize the system in a random product state $|\psi_{\rm ini}\rangle$ in the generalized Pauli $X$ basis, i.e., each spin is independently in $\{|0),|1),\ldots,|d-1)\}$ with equal probability. Let $\alpha>1$ and $\rho_A(t):=\tr_B(U(t,0)|\psi_{\rm ini}\rangle\langle\psi_{\rm ini}|U^\dag(t,0))$ be the reduced density matrix of subsystem $A$ at time $t$. If charge transport under the dynamics $U(t,0)$ is not faster than diffusive in the sense of Condition \ref{def:diff}, then
\begin{equation} \label{maineq}
R_\alpha(\rho_A)=\frac{\alpha}{\alpha-1}O\left(\sqrt{t\log t}\right)
\end{equation}
holds with probability $\ge1-1/p(t)$, where $p$ is a polynomial of arbitrarily high degree.
\end{theorem}

\begin{proof}
We divide the system into two subsystems $C\otimes D$. Subsystem $C$ consists of $m$ qudits at sites $N/2-m/2+1,N/2-m/2+2,\ldots,N/2+m/2$ near the cut, where $m$ is an even positive integer to be determined later. Subsystem $D$ is the rest of the system. The initial state can be expressed as
\begin{equation}
|\psi_{\rm ini}\rangle=|\psi_{\rm ini}\rangle_C\otimes|\psi_{\rm ini}\rangle_D,
\end{equation}
where $|\psi_{\rm ini}\rangle_C$ and $|\psi_{\rm ini}\rangle_D$ are random product states in subsystems $C$ and $D$, respectively. Define
\begin{equation}
|\psi_0\rangle=|0\rangle^{\otimes m}_C\otimes|\psi_{\rm ini}\rangle_D
\end{equation}
so that $|\langle\psi_0|\psi_{\rm ini}\rangle|=d^{-m/2}$. Since $U(t,0)$ is unitary,
\begin{equation}
|\langle U(t,0)\psi_0,U(t,0)\psi_{\rm ini}\rangle|=d^{-m/2}.
\end{equation}
The left-hand side of this equation is the absolute value of the inner product of $U(t,0)|\psi_0\rangle$ and $U(t,0)|\psi_{\rm ini}\rangle$. Occasionally we do not use the standard Dirac notation because it is cumbersome. Let $P:=|0\rangle_{N/2}\langle0|_{N/2}\otimes|0\rangle_{N/2+1}\langle0|_{N/2+1}$. Using (\ref{eq:diffc}) twice,
\begin{multline} \label{err}
\|(1-P)U(t,0)|\psi_0\rangle\|\le\|(1-|0\rangle_{N/2}\langle0|_{N/2})U(t,0)|\psi_0\rangle\|\\
+\||0\rangle_{N/2}\langle0|_{N/2}(1-|0\rangle_{N/2+1}\langle0|_{N/2+1})U(t,0)|\psi_0\rangle\|\le2e^{-\Omega(m^2/t)}.
\end{multline}

Assume without loss of generality that $N/2$ is odd. In $U(t,t-1)$, the only local unitary acting on both subsystems $A$ and $B$ is in the second product on the right-hand of Eq. (\ref{lu}). Define a modified local quantum circuit
\begin{align}
&V(t,0)=V(t,t-1)V(t-1,t-2)\cdots V(1,0),\\
&V(t,t-1)=\prod_{i=1}^{N/2-1} U_{2i,2i+1}^{(t)}\prod_{i=1}^{(N/2-1)/2}U_{2i-1,2i}^{(t)}u_{N/2,N/2+1}^{(t)}\prod_{i=(N/2+3)/2}^{N/2}U_{2i-1,2i}^{(t)},
\end{align}
where $u_{N/2,N/2+1}^{(t)}:=\langle00|U_{N/2,N/2+1}^{(t)}|00\rangle$ is a phase factor (complex number). It is easy to see that
\begin{equation}
U(t,t-1)P=V(t,t-1)P.
\end{equation}
Hence,
\begin{align}
U(t,0)|\psi_0\rangle&=U(t,t-1)U(t-1,0)|\psi_0\rangle\approx U(t,t-1)PU(t-1,0)|\psi_0\rangle\nonumber\\
&=V(t,t-1)PU(t-1,0)|\psi_0\rangle\approx V(t,t-1)U(t-1,0)|\psi_0\rangle,
\end{align}
where the error of each approximation step is upper bounded by (\ref{err}). Iterating this process,
\begin{equation}
\||\Delta_t\rangle\|=4te^{-\Omega(m^2/t)},\quad|\Delta_t\rangle:=U(t,0)|\psi_0\rangle-V(t,0)|\psi_0\rangle.
\end{equation}

Recall that both $|\psi_{\rm ini}\rangle_C$ and $|\psi_{\rm ini}\rangle_D$ are random product states in the generalized Pauli $X$ basis. We now fix the latter but not the former. Then, $|\psi_0\rangle$ and hence $|\Delta_t\rangle$ are fixed but $|\psi_{\rm ini}\rangle$ is not. Let
\begin{equation}
S=\{|0),|1),\ldots,|d)\}^{\otimes m}_C\otimes|\psi_{\rm ini}\rangle_D
\end{equation}
with $|S|=d^m$ be the set of all possible initial states consistent with $|\psi_{\rm ini}\rangle_D$. Since the states in $S$ are pairwise orthogonal,
\begin{equation}
\frac{1}{|S|}\sum_{|\psi_{\rm ini}\rangle\in S}|\langle\Delta_t|U(t,0)|\psi_{\rm ini}\rangle|^2\le d^{-m}\||\Delta_t\rangle\|^2.
\end{equation}
Define a subset of $S$ as
\begin{equation}
S':=\left\{|\psi_{\rm ini}\rangle\in S:|\langle\Delta_t|U(t,0)|\psi_{\rm ini}\rangle|\le d^{-m/2}\||\Delta_t\rangle\|\sqrt{p(t)}\right\}.
\end{equation}
Markov's inequality implies that
\begin{equation}
|S'|/|S|\ge1-1/p(t).
\end{equation}
It suffices to prove Eq. (\ref{maineq}) for all (initial) states in $S'$. To this end, we use

\begin{lemma} [Eckart-Young theorem \cite{EY36}] \label{EY}
Let
\begin{equation}
|\psi\rangle=\sum_{i\ge1}\lambda_i|a_i\rangle_A\otimes|b_i\rangle_B
\end{equation}
be the Schmidt decomposition of the state $|\psi\rangle$, where $\lambda_1\ge\lambda_2\ge\cdots>0$ with $\sum_{i\ge1}\lambda_i^2=1$ are the Schmidt coefficients in descending order. Then,
\begin{equation}
|\langle\phi|\psi\rangle|\le|\langle\psi'|\psi\rangle|=\sqrt{\sum_{i=1}^D\lambda_i^2}
\end{equation}
for any normalized state $|\phi\rangle$ of Schmidt rank $D$, where
\begin{equation}
|\psi'\rangle:=\frac{1}{\sqrt{\sum_{i=1}^D\lambda_i^2}}\sum_{i=1}^D\lambda_i|a_i\rangle_A\otimes|b_i\rangle_B.
\end{equation}
\end{lemma}

For any state $|\psi_{\rm ini}\rangle\in S'$,
\begin{align} \label{overlap}
&|\langle V(t,0)\psi_0,U(t,0)\psi_{\rm ini}\rangle|=|\langle U(t,0)\psi_0,U(t,0)\psi_{\rm ini}\rangle-\langle \Delta_t|U(t,0)|\psi_{\rm ini}\rangle|\nonumber\\
&\ge d^{-m/2}-|\langle\Delta_t|U(t,0)|\psi_{\rm ini}\rangle|\ge d^{-m/2}\left(1-\||\Delta_t\rangle\|\sqrt{p(t)}\right)\nonumber\\
&=d^{-m/2}\left(1-4te^{-\Omega(m^2/t)}\sqrt{p(t)}\right).
\end{align}
Let $\lambda_1$ be the largest Schmidt coefficient of $U(t,0)|\psi_{\rm ini}\rangle$, and $\Lambda_1=\lambda_1^2$ be the largest eigenvalue of the reduced density matrix $\rho_A(t)=\tr_B(U(t,0)|\psi_{\rm ini}\rangle\langle\psi_{\rm ini}|U^\dag(t,0))$. Since none of the local unitaries in $V(t,t-1)$ act on both subsystems $A$ and $B$, $V(t,t-1)$ and hence $V(t,0)$ do not generate entanglement so that $V(t,0)|\psi_0\rangle$ is a product state between $A$ and $B$, i.e., a state of Schmidt rank $1$. Combining this observation with (\ref{overlap}) and Lemma \ref{EY},
\begin{equation} \label{eq:f}
\lambda_1\ge d^{-m/2}\left(1-4te^{-\Omega(m^2/t)}\sqrt{p(t)}\right).
\end{equation}
Lemma \ref{l:1} implies that
\begin{equation}
R_\alpha(\rho_A)\le\frac{\alpha}{\alpha-1}R_\infty(\rho_A)=-\frac{\alpha}{\alpha-1}\ln\Lambda_1=-\frac{2\alpha}{\alpha-1}\ln\lambda_1.
\end{equation}
We complete the proof by letting $m=O(\sqrt{t\log t})$ with a sufficiently large coefficient hidden in the Big-O notation such that the factor in parentheses on the right-hand side of (\ref{eq:f}) is lower bounded by a positive constant.
\end{proof}

Combined with Lemma \ref{l:r}, the conclusion of Theorem \ref{thm} applies in particular to Haar-random local quantum circuits with charge conservation.

As shown in Ref. \cite{Hua19}, it is straightforward to extend Theorem \ref{thm} to local quantum circuits where charge transport is sub- or super-diffusive. It is also straightforward to extend Theorem \ref{thm} to two and higher spatial dimensions.

\section*{Acknowledgments}

I would like to thank Aram W. Harrow for his suggestions which helped to improve the presentation and Marko \v Znidari\v c for pointing out a minor error (which is more than a typo) in a draft of this note. This work was supported by NSF grant PHY-1818914 and a Samsung Advanced Institute of Technology Global Research Partnership.

\bibliographystyle{abbrv}
\bibliography{main}

\begin{thebibliography}{1}

\bibitem{EY36}
C.~Eckart and G.~Young.
\newblock The approximation of one matrix by another of lower rank.
\newblock {\em Psychometrika}, 1(3):211--218, 1936.

\bibitem{Hua19}
Y.~Huang.
\newblock Dynamics of {R}enyi entanglement entropy in local quantum circuits
  with charge conservation.
\newblock arXiv:1902.00977.

\bibitem{KVH18}
V.~Khemani, A.~Vishwanath, and D.~A. Huse.
\newblock Operator spreading and the emergence of dissipative hydrodynamics
  under unitary evolution with conservation laws.
\newblock {\em Physical Review X}, 8(3):031057, 2018.

\bibitem{RPv18}
T.~Rakovszky, F.~Pollmann, and C.~W. von Keyserlingk.
\newblock Diffusive hydrodynamics of out-of-time-ordered correlators with
  charge conservation.
\newblock {\em Physical Review X}, 8(3):031058, 2018.

\bibitem{RPv19}
T.~Rakovszky, F.~Pollmann, and C.~W. von Keyserlingk.
\newblock Sub-ballistic growth of {R}\'enyi entropies due to diffusion.
\newblock {\em Physical Review Letters}, 122(25):250602, 2019.

\bibitem{ZL20}
T.~Zhou and A.~W.~W. Ludwig.
\newblock Diffusive scaling of {R}\'enyi entanglement entropy.
\newblock {\em Physical Review Research}, 2(3):033020, 2020.

\bibitem{Zni20}
M.~{\v Z}nidari{\v c}.
\newblock Entanglement growth in diffusive systems.
\newblock {\em Communications Physics}, 3:100, 2020.

\end{thebibliography}

\end{document}